\documentclass[11pt,a4paper,fleqn]{scrartcl}
\usepackage{amssymb}
\usepackage{graphicx}
\usepackage{amsmath}
\usepackage{lineno}
\usepackage{xcolor}
\usepackage[normalem]{ulem}
\newtheorem{lemma}{Lemma}
\newtheorem{theorem}{Theorem}
\newtheorem{corollary}{Corollary}
\newtheorem{remark}{Remark}
\newtheorem{definition}{Definition}

\begin{document}

\title{Positive and Negative Determinant Strategies in Repeated Games with Behavior-Value Inconsistency}

\author{
	\small{
	Yuan Liu$^{1,2,3}$,
	Yakun Wang$^{1,2,4}$,
	Bin Wu$^{1,2,*}$
	}
}
\date{
	\footnotesize{
		\(^1\) School of Mathematical Sciences, Beijing University of Posts and Telecommunications, Beijing, China\\
		\(^2\) Key Laboratory of Mathematics and Information Networks (Beijing University of Posts and Telecommunications), Ministry of Education, Beijing, China\\
		\(^3\) Department of Theoretical Biology, Max Planck Institute for Evolutionary Biology, Plön, Germany\\
		\(^4\) Department of Engineering, Universitat Pompeu Fabra, Barcelona, Spain\\
		\(^*\) Bin Wu: bin.wu@bupt.edu.cn
	}
}

\maketitle

\begin{abstract}
Direct reciprocity, based on the repeated interactions, is a fundamental mechanism to promote cooperation.
Zero-determinant (ZD) strategies have opened an avenue for unilateral payoff control. 
However, previous studies neglect internal costs provided what agents do differ from what agents think, which is crucial for decision making of intelligent agents. 
Motivated by this, we establish a game theoretical framework by assuming that an individual pays the internal cost if the behavior is inconsistent with the internal thought. 
We prove that ZD strategy does not exist if the cost via behavior-value inconsistency is present.
Instead, we find a new class of repeated strategies that enforce a unilateral payoff control, which is termed as positive/negative determinant strategy. 
The found strategy allows an individual to enforce an affine combination of two individuals' average payoffs above/below zero.
Consequently, a focal individual is able to unilaterally control the opponent's payoff below a given value via negative determinant strategy, and a focal individual is able to get more payoff than the opponent via positive determinant strategy.
We also find that the control ability of positive/negative determinant strategies is better off than that of ZD strategies.
Our work highlights the importance of inconsistency between the behavior and value on payoff control, which is typically absent in classic ZD strategies. 
\end{abstract}

\section{Introduction}
\label{sec:introduction}
Cooperation is a non trivial phenomenon in the sight of evolutionary theory. Recent decades have seen a progress in the mechanisms to promote cooperation. 
Direct reciprocity, based on the repeated encounters, is a fundamental mechanism to promote cooperation \cite{Nowak_2006, Ohtsuki_2007, Nowak_1993, mailath2006repeated, aumann1995repeated}.  
The repeated Prisoner's Dilemma game is typically used to model the direct reciprocity.
Axelrod discovered the ``winning strategy" that an individual cooperates in the first round and repeats the opponent's behavior in the previous round, i.e., tit-for-tat \cite{Axelrod_1984}.
Tit-for-tat strategy has the property that if a player adopts such strategy, both players have the same payoff in the long run, regardless of the strategy chosen by the opponent.
Boerlijst et al. discovered a general class of such strategies named by equalizer strategies \cite{Boerlijst_1997}.
In 2012, Press and Dyson discovered a more general class of strategies, which allow a player to enforce an affine relationship between her own payoff and the opponent’s payoff, regardless of the opponent’s strategy \cite{Press_2012}.
This is referred as zero-determinant strategies (ZD strategies), since they are obtained via the manipulation of matrices without changing its determinant. 

ZD strategy opens an avenue for unilateral payoff control \cite{akin2016iterated, Hao_2015, pan2015zero, mamiya2020zero, BOERLIJST1997281, chen2023outlearning, ueda2021memory, adami2013evolutionary, cheng2024design}. 
Using ZD strategies, the focal individual is able to enforce an extortionate share of payoffs.
Furthermore, the control can be unilateral. Last but not the least, the unilateral payoff control via ZD strategy makes use of almost the least information, that is, only the focal and opponent’s behavior in the current round without payoff evaluation and without recursively updating one’s response. It is a once-for-all strategy designed at the beginning of the game.
Beyond expectations, such counterintuitive unilateral payoff control is ubiquitous ranging from pairwise two-strategy games to multi-player multi-strategy games \cite{Govaert2021, Ueda_2022, hilbe2017memory, liu2022environmental, GOVAERT2019150}. 
Using the payoff control framework, a controller can restrict the relation between her and the opponent’s payoffs to an arbitrary region with linear boundaries, as long as the control objective is feasible. 
To be more more precisely, a single player can (i) unilaterally determine the maximum and minimum values of the opponent’s possible payoffs; or (ii) always win the game no matter what the opponent’s strategy is; or (iii) control the evolutionary route of the game, as long as the opponent is rational and self-optimizing \cite{hao2018payoff}.
The controller can enforce the game to finally converge either to a mutual-cooperation equilibrium or to any feasible equilibrium that she wishes.
As a result, it achieves a variety of control goals, including:
adjusting the payoff relationships between individuals to bring the system to a stable state;
influencing the collective welfare of the entire system through appropriate strategy adjustments;
adjusting strategies appropriately to enhance the system's adaptability to external disturbances or changes.

Payoff control has recently received increasingly attention in artificial intelligence (AI) \cite{hao2018payoff, tang2024incentive}. 
Previous payoff control strategies are only to be found in the classical games, where action gives rise to the entire strategy set. 
In artificial intelligence, agents are likely to adjust their actions (strategies) based on past experiences, typically using reinforcement learning to handle complex environments and strategic opponents. 
It is typically assumed that one agent is the payoff controller, using a pre-determined payoff control strategy, while the other is the reinforcement learner, adjusting their strategy according to reinforcement learning dynamics \cite{shi2025payoff}. 
It mirrors the interaction between humans and robots. The control question here is how to manipulate the robot to have a lower payoff than humans.
Intelligent decision makers including AI are likely to have a value inside, such as emotion \cite{assunccao2022overview}, intention \cite{zhao2024learning}, ethical implications \cite{elendu2023ethical} and so on, which can be crucial for decision making. 

Internal thought is a latent state or unobservable variable that influences a player's decision-making process, which classical decision making theory neglects. 
Internal thoughts are shaped by previous behaviors of both players and the focal individual’s internal thoughts. 
This idea has been present in collective decision-makings, where individuals' decisions are not made in isolation; but evolve in response to the actions and opinions of others \cite{zino2020two, martins2008continuous, gargiulo2012influence}. 
An individual's opinion evolves not only via sharing opinions with others but also by observing the actions of those around them.
This highlights the crutial role of internal thought in shaping decision-making. 
As Charles Horton Cooley says \cite{cooley2017human}, ``Each to each a looking-glass reflects the other that doth pass." 
Thus, the internal thought is indeed part of the strategy, influenced by the behaviors they observe.
However, this value-behavior framework is absent in classical game theory.

Motivated by this, we propose that behavior of both individuals influence the internal thought in the future. 
In this case, individuals whose action differs from its value have some extra costs, say psychological costs. 
The value-behavior inconsistency reflects cognitive dissonance, a phenomenon widely observed in  psychological experiments \cite{Asch_1956, Schwarz_2000, 10.2307/24936719}. 
Previous psychological game theory suggests that participants may gain additional psychological effects from their emotions, which can be fully reflected by the payoffs in the game. 
The natural question is whether these psychological effects are captured by simply adjusting the payoff in classical game theory \cite{elster1998emotions}? 
It does not work. The values are endogenous. 
Therefore, merely modifying the payoff function without modeling the dynamics of internal values fails to explain value-behavior inconsistency.
These internal costs are not reflected in static payoffs, in which preferences are fixed \cite{geanakoplos1989psychological}. 
Thus, it is necessary to consider the individual’s value by changing the structure of the game, allowing internal values and psychological mechanisms to be explicitly represented and dynamically modeled. 
For this reason, we assume that each individual holds its internal thoughts (cooperation or defection), which incurs an additional psychological cost (internal cost) provided that its behavior deviates from the internal thought. This modeling captures the psychological effects of the internal cost by generalizing the classical prisoner’s dilemma game to a new two-player four-strategy game. 
How does this inner cost alter the payoff controllability? The AI-human interactions are likely to be more frequent in the near future with AI more and more engaged in control systems. A payoff control with value-behavior inconsistency is thus necessary to investigate such human-AI hybrid systems. 
To this end, our objective is to explore whether unilateral payoff control remains effective in such intricate behavior-value games and, if so, to what extent.

This manuscript is organized as follows. In Section \ref{sec2}, a model of two-player two-strategy repeated game with behavior-value inconsistency is proposed.
In Section \ref{sec3}, the existence of positive/negative determinant strategy is proved. 
In Section \ref{sec6}, some examples are shown in systems with control background.
We conclude this article in Section \ref{sec5}.

\section{Model}
\label{sec2}
Let us consider a two-player symmetric repeated game. 
In contrast with classical games, we introduce the inner state or value, giving rise to a game 
$G = \{N, \{ \mathcal{B}^i\} _{i \in N},  \{  \mathcal{V}^i \} _{i \in N},$  $\{ S^i \} _{i \in N} \},$ where $N=\{1, 2\}$ represents the set of players. 
Each individual $i \in N$ has a set of behaviors $\mathcal{B}^i$, including $C$ (cooperation, $\mathcal{B}^i_1$) and $D$ (defection, $ \mathcal{B}^i_2$); and a set of values $\mathcal{V}^i$, including $\mathbf{C}$ (cooperation, $\mathcal{V}^i_1$) and $\mathbf{D}$ (defection, $\mathcal{V}^i_2$). 
The payoff function $S^i : \bold{\mathcal{B}} \times \mathcal{V}^i \to \mathcal{R}$, in which $\bold{\mathcal{B}} := \prod_{i \in N} \mathcal{B}^i$, determines the payoff received by each individual. 
Here, $\mathcal{R}$ represents the set of real numbers.

Both individuals have four strategies $C\mathbf{C}$, $C\mathbf{D}$, $D\mathbf{C}$, $D\mathbf{D}$.
Either individual can serve as the focal individual.
Without loss of generality, we choose individual $1$ as the focal individual. 
If the internal thought of the focal individual is defection, 
i) she gets $R-\varepsilon_1$ (reward) when both individuals choose to cooperate;
ii) she gets $S-\varepsilon_1$ (sucker payoffs) when she is a cooperator and her opponent is a defector;
iii) she gets $T$ (temptation) when she is a defector and her opponent is a cooperator;
iv) she gets $P$ (punishment) when both individuals choose to defect.
Thus, $\varepsilon_1$ is the cost of having behavior $C$ and internal thought $\mathbf{D}$. 
Analogously, individual pays the internal cost $\varepsilon_2$, if the individual's behavior is defection but the internal thought is cooperation. 
Then the payoff matrix is given by
\begin{equation}
\label{game1}
\bordermatrix{
		     & C\mathbf{C}       &  C\mathbf{D} &  D\mathbf{C} & D\mathbf{D} \cr
 C\mathbf{C} &  R			 & 	R		&S&S\cr
 C\mathbf{D} & R-\varepsilon_1 & R-\varepsilon_1 & S -\varepsilon_1 & S -\varepsilon_1\cr 
 D\mathbf{C} & T-\varepsilon_2 & T-\varepsilon_2 & P-\varepsilon_2 & P-\varepsilon_2\cr 
 D\mathbf{D} & T			 & 	T		&P&P}.
\end{equation}
We assume $T > R > P > S$ and $\varepsilon_1, \varepsilon_2  >0$.

As time evolves, we have a repeated game. 
Here we assume individual $i \in N$ cannot see the internal thought of opponent, that is, she only knows three elements of the current state $\mathcal{B}^1,\mathcal{B}^2,\mathcal{V}^i$. 
Each individual $i \in N$ adopts memory-one strategy (Markovian property), implying that the strategy in the current round depends on $(\mathcal{B}^1 (t - 1),\mathcal{B}^2 (t - 1),\mathcal{V}^i (t - 1))$ from the previous round. 
It leads to $(\bold{p}^i_{\mathcal{B}} (t) ,\bold{q}^i_{\mathcal{V}} (t))$, where 
\begin{equation}
\begin{split}
\bold{p}^i_{\mathcal{B}} (t): \bold{\mathcal{B}} (t - 1) \times \mathcal{V}^i (t - 1) \to \mathcal{B}^i (t),\\
\bold{q}^i_{\mathcal{V}} (t): \bold{\mathcal{B}} (t - 1) \times \mathcal{V}^i (t - 1) \to \mathcal{V}^i (t).
\end{split}
\end{equation}
Here, $\bold{\mathcal{B}} (t - 1) := \prod_{i \in N} \mathcal{B}^i (t - 1), t >1$.

We assume that the strategies remain fixed throughout the entire evolutionary process. 
Specifically, for any time any time $t > 1$, we have $\bold{p}^i_{\mathcal{B}} (t) = \bold{p}^i_{\mathcal{B}}$, and $\bold{q}^i_{\mathcal{V}} (t) = \bold{q}^i_{\mathcal{V}}$.
The strategies of individual $i$, $\bold{p}^i_{\mathcal{B}}$ and $\bold{q}^i_{\mathcal{V}}$, are given by
\begin{equation}
\begin{split}
\bold{p}^i_{\mathcal{B}} = (&h^i_{C|CC\mathbf{C}}, h^i_{C|CD\mathbf{C}}, h^i_{C|DC\mathbf{C}}, 
\dots, h^i_{C|DD\mathbf{D}}),\\
\bold{q}^i_{\mathcal{V}} = (&g^i_{\mathbf{C}|CC\mathbf{C}}, g^i_{\mathbf{C}|CD\mathbf{C}}, g^i_{\mathbf{C}|DC\mathbf{C}}, 
 \dots, g^i_{\mathbf{C}|DD\mathbf{D}}).
\end{split}
\end{equation}
The entries $h ^ i _ {C| \mathcal{B}^1 \mathcal{B}^2 \mathcal{V}^ i}$ and $g ^ i _ {\mathbf{C}| \mathcal{B}^1 \mathcal{B}^2 \mathcal{V}^ i}$ are conditional probabilities that the behavior and the internal thought of the individual $i (i \in N)$ are cooperation in the next round, respectively, provided that the current state is ${\mathcal{B}^1\mathcal{B}^2\mathcal{V}^i}$, regardless of the internal thought of the opponent.
We abbreviate it as
\begin{equation}
\begin{split}
\bold{p}^i_{\mathcal{B}} = (&p^i_1, p^i_2, \dots, p^i_8),\\
\bold{q}^i_{\mathcal{V}} = (&q^i_1, q^i_2, \dots, q^i_8).
\end{split}
\end{equation}

The states of the individuals (i.e., behavior and internal thought, $\mathcal{B}^1 \mathcal{B}^2 \mathcal{V}^1 \mathcal{V}^2$) evolve over discrete time.
There are sixteen states in the state space $\mathcal{S}=$ $\{C, D\} ^ 2 \times \{ \mathbf{C}, \mathbf{D} \} ^2$.
We have $h^1_{C|\mathcal{B}^1 \mathcal{B}^2 \mathcal{V}^1 \mathbf{C}} 
= h^1_{C|\mathcal{B}^1 \mathcal{B}^2 \mathcal{V}^1 \mathbf{D}}$, 
$g^1_{C|\mathcal{B}^1 \mathcal{B}^2 \mathcal{V}^1 \mathbf{C}} 
= g^1_{C|\mathcal{B}^1 \mathcal{B}^2 \mathcal{V}^1 \mathbf{D}}$, 
$h^2_{C|\mathcal{B}^1 \mathcal{B}^2 \mathbf{C} \mathcal{V}^2} 
= h^2_{C|\mathcal{B}^1 \mathcal{B}^2 \mathbf{D} \mathcal{V}^2}$, and
$g^2_{C|\mathcal{B}^1 \mathcal{B}^2 \mathbf{C} \mathcal{V}^2} 
= g^2_{C|\mathcal{B}^1 \mathcal{B}^2 \mathbf{D} \mathcal{V}^2}$
The probability that they are still cooperative in their behavior and internal thought in the next round is given by $p^1_1 p^2_1 q^1_1 q^2_1$, if both individuals are cooperative in their behavior and thought in the previous round (See Fig.\ref{model}).
\begin{figure}
	\begin{minipage}{\linewidth} 
	\centerline{\includegraphics[scale=0.45]{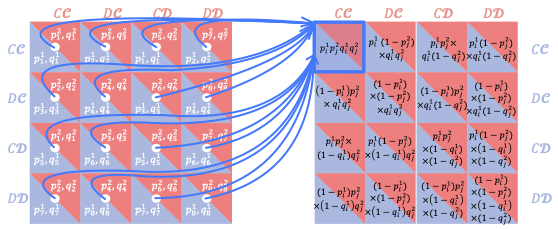}}
	\end{minipage}
	\caption{\textbf{Complexity of strategy set arising from the inconsistency between behavior and value.} 
	Each individual is characterized by a two-dimensional state variable $(\mathcal{B}, \mathcal{V})$. 
$\mathcal{B}$ represents a set of behavior, including $C$ (cooperation) and $D$ (defection); 
and $\mathcal{V}$ is the value that one holds (the internal thought), including $\mathbf{C}$ (cooperation) and $\mathbf{D}$ (defection).
There are sixteen states in the state space $\mathcal{S}=$ $\{C, D\} ^ 2 \times \{ \mathbf{C}, \mathbf{D} \} ^2$.
Both individuals take memory-one strategies and cannot see the internal thoughts of their opponent. 
Each individual's strategy is represented by two $8$-dimensional vectors, denoted as $(\bold{p}^i_{\mathcal{B}}, \bold{q}^i_{\mathcal{V}})$ for individual $i$ ($i = 1, 2$). 
Each entry in $\bold{p}^i_{\mathcal{B}}$ ($\bold{q}^i_{\mathcal{V}}$) represents the conditional probability that the behavior (the internal thought) of individual $i$ is cooperation in the next round, based on the $16$ possible outcomes of the previous move. This probability is determined irrespective of the internal thoughts of their opponent.
	}\label{model} 
\end{figure} 
Thus both individual behavior and internal thought are Makovian, which yields a Markovian eco-evolutionary dynamics. 
The transition matrix $\mathbb{P}$ is a square probability matrix of order $16$.
For non-deterministic strategies, (any entry of $\bold{p}^i_{\mathcal{B}}$, $\bold{p}^i_{\mathcal{V}}$ ($i = 1,2$) is neither zero nor one), the Markov chain is aperiodic and irreducible \cite{karlin2014first, karlin1981second}.
This implies that there is a unique stationary distribution, which is determined by the left eigenvector $\bold{v}$ of the unit eigenvalue, i.e., $\bold{v} \mathbb{P} =\bold{v}$.

\section{Existence of Positive and Negative Determinant Strategies}
\label{sec3}
\begin{definition}
\label{def.1}
(Average payoffs of two players)
The average payoffs of individuals $1$ and $2$ are 
\begin{equation}\label{payoff1}
\begin{split}
s^1 &= \bold{v} \cdot \bold{S}^1,\\ 
s^2 &= \bold{v} \cdot \bold{S}^2.
\end{split}
\end{equation}
The vector $\bold{v}$ is the stationary distribution of 16-dimension Markov transition matrix $\mathbb{P}$, $\bold{S}^1$ and $\bold{S}^2$ are payoff vectors of individuals $1$ and $2$ given by 
\begin{equation}
\begin{split}\label{eq8}
\bold{S}^1=(&R, S, T - \varepsilon_2, P - \varepsilon_2, R - \varepsilon_1, S - \varepsilon_1, T, P, \\
&R, S, T - \varepsilon_2, P - \varepsilon_2, R - \varepsilon_1, S - \varepsilon_1, T, P)^T,\\
\bold{S}^2=(&R, T - \varepsilon_2, S, P - \varepsilon_2,R, T - \varepsilon_2, S, P - \varepsilon_2, \\
&R - \varepsilon_1, T, S - \varepsilon_1, P, R - \varepsilon_1, T, S - \varepsilon_1, P)^T.
\end{split}
\end{equation}
The first to the last entry of $\bold{S^1}$ ($\bold{S^2}$) represents the payoff of individual $1$ (individual $2$) in state
$CC\mathbf{C}\mathbf{C}$, $CD\mathbf{C}\mathbf{C}$, $DC\mathbf{C}\mathbf{C}$, $DD\mathbf{C}\mathbf{C}$, $CC\mathbf{D}\mathbf{C}$, $CD\mathbf{D}\mathbf{C}$, $DC\mathbf{D}\mathbf{C}$, $DD\mathbf{D}\mathbf{C}$, $CC\mathbf{C}\mathbf{D}$, $CD\mathbf{C}\mathbf{D}$, $DC\mathbf{C}\mathbf{D}$, $DD\mathbf{C}\mathbf{D}$, $CC\mathbf{D}\mathbf{D}$, $CD\mathbf{D}\mathbf{D}$, $DC\mathbf{D}\mathbf{D}$, $DD\mathbf{D}\mathbf{D}$ $\in \mathcal{S}$, respectively.
For example, individual $1$ gets $S$ (i.e., $\bold{S^1}_{CD\mathbf{C}\mathbf{C}}$) and individual $2$ gets $T - \varepsilon_2$ (i.e., $\bold{S^2}_{CD\mathbf{C}\mathbf{C}}$), if the behavior of individuals $1$ and $2$ are $C$ and $D$ and the internal thought of individuals $1$ and $2$ are $\mathbf{C}$. 
\end{definition}

\begin{remark}
\label{remm}
The average payoff of individual $i (i = 1,2)$ $s^i$ is calculated as the ratio of the accumulated payoffs to the number of rounds.
If the game is played infinitely many rounds, the average payoff $s^i$ is defined as:
\begin{equation}\label{payoff11}
s^i = {\lim_{T \to \infty} \frac{1}{T} \sum_{t = 1}^T S^i(t)}.
\end{equation}
Here $T$ represents the number of rounds, and $S^i(t)$ is the payoff of individual $i$ in round $t$. 
The average payoff is the long-term payoff a player receives per round as the number of rounds approaches infinity. 
For non-deterministic strategies, there is a unique stationary distribution $\bold{v}$ of the transition matrix $\mathbb{P}$, which is determined by $\bold{v} \mathbb{P} =\bold{v}$.
Notably, all entries of the stationary distribution $\bold{v}$ are strictly positive, with $0 < \bold{v}_i < 1$ for all $i=1,2,\dots,16$, and satisfy $\sum_i \bold{v}_i = 1$.
Thus, Eq.\ref{payoff11} is equivalent to Eq.\ref{payoff1}.
\end{remark}

\begin{definition}
\label{def.2}
(Positive/Zero/Negative determinant strategies)
A positive determinant strategy of individual $1$ $(\bold{p}^1_{\mathcal{B}}, \bold{p}^1_{\mathcal{V}})$ is a strategy that satisfies the following condition: For any strategies adopted by the player $2$, there exist constants $\alpha, \beta, \gamma$ (not all $0$) which are unilaterally decided by player $1$, such that two players’ average payoffs fulfill: 
	\begin{equation}
	\alpha s^1 + \beta s^2 +\gamma > 0,
	\end{equation}
where $s^1$ and $s^2$ are average payoffs of individuals $1$ and $2$.

A zero determinant strategy of individual $1$ $(\bold{p}^1_{\mathcal{B}}, \bold{p}^1_{\mathcal{V}})$ is a strategy that satisfies the following condition: For any strategies adopted by the player $2$, there exist constants $\alpha, \beta, \gamma$ (not all $0$) which are unilaterally decided by player $1$, such that two players’ average payoffs fulfill: 
	\begin{equation}
	\alpha s^1 + \beta s^2 +\gamma = 0.
	\end{equation}

A negative determinant strategy of individual $1$ $(\bold{p}^1_{\mathcal{B}}, \bold{p}^1_{\mathcal{V}})$ is a strategy that satisfies the following condition: For any strategies adopted by the player $2$, there exist constants $\alpha, \beta, \gamma$ (not all $0$) which are unilaterally decided by player $1$, such that two players’ average payoffs fulfill: 
	\begin{equation}
	\alpha s^1 + \beta s^2 +\gamma < 0.
	\end{equation}
	
\end{definition}

\begin{remark}
Individuals only know the likelihood to do cooperation based on the behaviors and thoughts of her own and the behavior of the opponent. 
The strategy $(\bold{p}^1_{\mathcal{B}}, \bold{p}^1_{\mathcal{V}})$ is fixed over time.
They do not adjust their behavior/thought based on future expected payoffs, as in reinforcement learning \cite{kaelbling1996reinforcement, ernst2024introduction}.
\end{remark}

\begin{definition}
\label{def.3}
(Press-Dyson vectors \cite{akin2016iterated, Ueda_2022})
For the memory-one strategies of player $i$ ($(\bold{p}^i_{\mathcal{B}}, \bold{p}^i_{\mathcal{V}})$, $i = 1, 2$), the Press–Dyson vectors $(\tilde {\bold{p}^i_{\mathcal{B}}}, \tilde {\bold{p}^i_{\mathcal{V}}})$ are 
\begin{equation}\label{pdvector}
\begin{split}
&\tilde {p^i_{\mathcal{B}}}(C|\mathcal{B}^1 \mathcal{B}^ 2 \mathcal{V}^ 1 \mathcal{V}^ 2) = h ^ i _ {C| \mathcal{B}^1 \mathcal{B}^ 2 \mathcal{V}^ 1 \mathcal{V}^ 2} - \delta ^ i _ {C , \mathcal{B}^ i},\\
&\tilde {q^i_{\mathcal{V}}}(\mathbf{C}|\mathcal{B}^1 \mathcal{B}^ 2 \mathcal{V}^ 1 \mathcal{V}^ 2) = g ^ i _ {\mathbf{C}|\mathcal{B}^1 \mathcal{B}^ 2 \mathcal{V}^ 1 \mathcal{V}^ 2} - \delta ^ i _ {\mathbf{C}, \mathcal{V}^ i},
\end{split}
\end{equation}
where $\delta ^ i _ {C , \mathcal{B}^ i}$ and $\delta ^ i _ {\mathbf{C}, \mathcal{V}^ i}$ are the Kronecker delta. 
The second terms on the right-hand side of Eqs.\ref{pdvector} $\delta ^ i _ {C , \mathcal{B}^ i} (\delta ^ i _ {\mathbf{C}, \mathcal{V}^ i})$ can be regarded as the strategy “Repeat”, which repeats her own behavior/internal thought in the previous round. 
The Press-Dyson vectors are thus regarded as the difference between her own strategy and “Repeat”.
\end{definition}

\begin{lemma}
\label{lem1}
(Existence condition of the zero determinant strategy \cite{Ueda_2022})
For any $\mathcal{B}^ i \in \{ C, D\}, \mathcal{V}^ i \in \{ \mathbf{C}, \mathbf{D}\}$, if there exist constants $\alpha, \beta, \gamma$ which are unilaterally decided by the focal individual $1$,  such that the Press–Dyson vectors satisfy:
\begin{equation}\label{spl13}
a \tilde {\bold{p}^1_{\mathcal{B}}} + b \tilde {\bold{p}^1_{\mathcal{V}}}= \alpha \bold{S}^1 + \beta \bold{S}^2 + \gamma \bold{1},
\end{equation}
$a$ and $b$ are not all zero, $\bold{1}$ is the vector with all components $1$,
then the strategy $(\bold{p}^1_{\mathcal{B}}, \bold{p}^1_{\mathcal{V}})$ is a zero determinant strategy.
\end{lemma}

To give our main result, we introduce some notations:
\begin{equation*}
\begin{split}
\bold{U} :=&\alpha \bold{S}^1 + \beta \bold{S}^2 +\gamma \bold{1},\\
\bold{A_1} :=& (U_1, U_2, U_3, U_4, U_5, U_6, U_7, U_8)^T,\\
\bold{A_2} := &(U_9, U_{10}, U_{11}, U_{12}, U_{13}, U_{14}, U_{15}, U_{16})^T,\\
\bold{B_1} := &(-\beta \varepsilon_1, \beta \varepsilon_2 , -\beta \varepsilon_1 , \beta \varepsilon_2 , -\beta \varepsilon_1 ,\beta \varepsilon_2 , -\beta \varepsilon_1 , \beta \varepsilon_2)^T,\\
\bold{B_2}  := &(\beta \varepsilon_1, -\beta \varepsilon_2, \beta \varepsilon_1, -\beta \varepsilon_2 ,\beta \varepsilon_1, -\beta \varepsilon_2 , \beta \varepsilon_1, -\beta \varepsilon_2)^T,\\
\bold{O} := &(0, 0, 0, 0, 0, 0, 0, 0)^T.\\
\end{split}
\end{equation*}

\begin{theorem}
	\label{thm1}
	(Existence condition of the positive/negative determinant strategy)
	For the two-player two-strategy repeated game given by Eq.\ref{game1}, if the following two conditions hold simultaneously \\
	i) there exists a sufficiently large constant $C > 0$, such that $\lvert \varepsilon_1 -\varepsilon_2 \rvert \ge C$,\\
	ii) for any $\mathcal{B}^ i \in \{ C, D\}, \mathcal{V}^ i \in \{ \mathbf{C}, \mathbf{D}\}$, there exist $a$ and $b$, not all zero, such that the Press–Dyson vectors satisfy:
\begin{gather}
\qquad a \tilde {\bold{p}^1_{\mathcal{B}}} + b \tilde {\bold{p}^1_{\mathcal{V}}}= \bold{A_1}, \label{eq.14}\\
\text{or\quad}a \tilde {\bold{p}^1_{\mathcal{B}}} + b \tilde {\bold{p}^1_{\mathcal{V}}}= \bold{A_2} \label{eq.15},
\end{gather}
then a positive or a negative determinant strategy $(\bold{p}^1_{\mathcal{B}}, \bold{p}^1_{\mathcal{V}})$ of player $1$ exists.
\end{theorem}

\noindent \textbf{Proof.}
Based on condition ii), there are two ways to split $\bold{U}$ (an affine combination of payoff vectors ($\bold{S}^1$, $\bold{S}^2$)) into two parts:
\begin{equation}\label{U1}
\bold{U}=\alpha \bold{S}^1 + \beta \bold{S}^2 +\gamma \bold{1}=
	\begin{pmatrix}
	\bold{A_1}  \\
	\bold{A_1} 
	\end{pmatrix}
	 +
	\begin{pmatrix}
	\bold{O}  \\
	\bold{B_1} 
	\end{pmatrix},
\end{equation}
and 
\begin{equation}\label{U2}
\bold{U}=\alpha \bold{S}^1 + \beta \bold{S}^2 +\gamma \bold{1}=
	\begin{pmatrix}
	\bold{A_2}  \\
	\bold{A_2} 
	\end{pmatrix}
	 +
	\begin{pmatrix}
	\bold{B_2}  \\
	\bold{O} 
	\end{pmatrix}.
\end{equation}
Here, we show the proof of the existence of positive/negative determinant strategy under the first way, i.e., Eq.\ref{U1}.
If Eq.\ref{U2} is fulfilling, the proof can be given in a similar way.

\textbf{First, we split any affine combination of average payoffs into two parts.}
We consider
\begin{equation}\label{spl15}
\bold{U}=\alpha \bold{S}^1 + \beta \bold{S}^2 +\gamma \bold{1}=
\underbrace{	
	\begin{pmatrix}
	\bold{A_1}  \\
	\bold{A_1} 
	\end{pmatrix} }_{\bold{f_1}}
	 +
\underbrace{	
	\begin{pmatrix}
	\bold{O}  \\
	\bold{B_1} 
	\end{pmatrix} }_{\bold{f_2}}.
\end{equation}
Based on Definition.\ref{def.1}, the average payoffs of individuals $1$ and $2$ are 
$
s^1 = \bold{v} \cdot \bold{S}^1,
s^2 = \bold{v} \cdot \bold{S}^2.
$
For any affine combination of payoffs, we have
\begin{equation}
\alpha s^1 + \beta s^2 +\gamma =\alpha \bold{v} \cdot  \bold{S}^1 + \beta \bold{v} \cdot  \bold{S}^2 + \gamma \bold{v} \cdot  \bold{1}
=\bold{v} \cdot (\alpha \bold{S}^1 + \beta \bold{S}^2 +\gamma \bold{1}).
\end{equation}
Based on Eq.\ref{spl15}, we have
\begin{equation}
\label{payoff}
\alpha s^1 + \beta s^2 +\gamma
=\bold{v} \cdot (\bold{f_1} + \bold{f_2})
=\bold{v} \cdot \bold{f_1} + \bold{v} \cdot  \bold{f_2}.
\end{equation}

\textbf{Second, we show that the sign of the affine combination of average payoffs is determined by $\beta$, $\varepsilon_1$ and $\varepsilon_2$.}
Based on condition ii), there exist $a$ and $b$ which are not all zero, such that Press-Dyson vectors $(\tilde {\bold{p}^i_{\mathcal{B}}}, \tilde {\bold{p}^i_{\mathcal{V}}})$ satisfy
\begin{equation}
a \tilde {\bold{p}^1_{\mathcal{B}}} + b \tilde {\bold{p}^1_{\mathcal{V}}} = \bold{f_1},
\end{equation}
thus $\bold{v} \cdot \bold{f_1} = 0$ holds according to Lemma.\ref{lem1}.

Here, $v_i$ is the $i_{th}$ entry of the stationary distribution $\bold{v}$. If both individuals adopt non-deterministic strategies, then $v_i > 0$ and $\sum_{i=1}^{16} v_i = 1$.
We have 
\begin{equation*}
\begin{split}
\bold{v} \cdot \bold{f_2}
=&\sum_{i=1}^8 0
+ \beta \varepsilon_2 \sum_{i=10, 12, 14, 16} v_i - \beta \varepsilon_1 \sum_{i=9, 11, 13, 15} v_i\\
=& \beta \varepsilon_2 \sum_{i=10, 12, 14, 16} v_{i} - \beta \varepsilon_1 \sum_{i=9, 11, 13, 15} v_{i}.
\end{split}
\end{equation*}
Thus, 
\begin{equation}
\label{pnstra}
\alpha s^1 + \beta s^2 + \gamma=
 \beta \varepsilon_2 \sum_{i=10, 12, 14, 16} v_{i} - \beta \varepsilon_1 \sum_{i=9, 11, 13, 15} v_{i}.
\end{equation}

\textbf{Finally, we give the existence conditions of the positive/negative determinant strategy.}
Therefore,
i) if $\beta > 0$, and there exists a sufficiently large constant $C > 0$, such that $\varepsilon_2 -\varepsilon_1  \ge C$, then $\alpha s^1 + \beta s^2 + \gamma>0$.
The memory-one strategy $(\bold{p}^1_{\mathcal{B}}, \bold{p}^1_{\mathcal{V}})$ of player $i$ is a positive determinant strategy.\\
ii) if $\beta > 0$, and there exists a sufficiently large constant $C > 0$, such that $\varepsilon_1 -\varepsilon_2  \ge C$, then $\alpha s^1 + \beta s^2 + \gamma < 0$.
The memory-one strategy $(\bold{p}^1_{\mathcal{B}}, \bold{p}^1_{\mathcal{V}})$ of player $i$ is a negative determinant strategy.\\
iii) if $\beta < 0$, and there exists a sufficiently large constant $C > 0$, such that $\varepsilon_2 -\varepsilon_1  \ge C$, then $\alpha s^1 + \beta s^2 + \gamma < 0$.
The memory-one strategy $(\bold{p}^1_{\mathcal{B}}, \bold{p}^1_{\mathcal{V}})$ of player $i$ is a negative determinant strategy.\\
iv) if $\beta < 0$, and there exists a sufficiently large constant $C > 0$, such that $\varepsilon_1 -\varepsilon_2  \ge C$, then $\alpha s^1 + \beta s^2 + \gamma > 0$.
The memory-one strategy $(\bold{p}^1_{\mathcal{B}}, \bold{p}^1_{\mathcal{V}})$ of player $i$ is a positive determinant strategy.

\begin{remark}
\textbf{Role of Lemma.\ref{lem1} in the proof of Theorem.\ref{thm1}.}
Although Theorem.\ref{thm1} is on the existence of a positive or negative determinant strategy, we employ Lemma.\ref{lem1} —originally derived for the zero-determinant (ZD) strategy—as a key step in the proof.
The main novelty lies in decomposing an affine relationship $\bold{U} =\alpha \bold{S}^1 + \beta \bold{S}^2 +\gamma \bold{1}$ into two parts (See Eqs.\ref{U1}, \ref{U2}), as $\bold{U} \notin \operatorname{span} \{\tilde {p^1_{\mathcal{B}}}, \tilde {p^1_{\mathcal{V}}} \}$. 
The first part (denoted as $\bold{f_1}$) is expressed as an affine combination of Press-Dyson vectors (See Eqs.\ref{eq.14}, \ref{eq.15}), leading to a zero determinant (via Lemma.\ref{lem1}). 
This part corresponds to a ZD-like construction.
The second part ($\bold{f_2}$) represents the residual component. 
It determines the sign of the determinant, which is influenced by the parameters $\varepsilon_1$, $\varepsilon_2$, and $\beta$. 
In particular, the term is not vanishing if the psychological costs are present and not equal with each other.
This demonstrates the existence of a positive or negative determinant strategy.
\end{remark}

\begin{remark}
If the conditions of Theorem.\ref{thm1} are satisfied, a positive (negative) determinant strategy exists based on the split of the affine relationship $\bold{U}=\alpha \bold{S}^1 + \beta \bold{S}^2 +\gamma \bold{1}$, that is, Eqs.\ref{eq.14}, \ref{eq.15}.
On the other hand, if we change the factors $a$ and $b$, there can be countless positive (negative) determinant strategies.
\end{remark}

\begin{remark}
The proof is valid for a game with $n$ players, $m$ strategies, $s$ internal thoughts. 
The dimension of the transition probability matrix is $m^n s^n$ in this general case.
\end{remark}

\begin{theorem}\label{thm2}
(The equivalent form of Theorem.\ref{thm1})
For the two-player two-strategy repeated game given by Eq.\ref{game1}, if\\
i) there exists a sufficiently large constant $C > 0$, such that $\lvert \varepsilon_1 -\varepsilon_2 \rvert \ge C$,\\
	ii) there exist constants $\alpha, \beta, \gamma$, not all zero, which are unilaterally decided by the focal player $1$, such that $\forall \mathcal{B}^2 \in \{ C, D \}, \mathcal{V}^2 \in \{ \mathbf{C}, \mathbf{D} \}$,
one of the following conditions is satisfied: 
\begin{small}
\begin{equation}
\label{inl.1}
\begin{split}
(a) \ &U_{C \mathcal{B}^2 \mathbf{C} \mathcal{V}^2} \ge 0,
	    U_{D \mathcal{B}^2 \mathbf{D} \mathcal{V}^2} \leq 0,
	 U_{D \mathcal{B}^2 \mathbf{C} \mathcal{V}^2} \ge 0,
	 U_{C \mathcal{B}^2 \mathbf{D} \mathcal{V}^2} \leq 0, \\
(b) \ &U_{C \mathcal{B}^2 \mathbf{C} \mathcal{V}^2} \ge 0,
	 U_{D \mathcal{B}^2 \mathbf{D} \mathcal{V}^2} \leq 0,
	 U_{D \mathcal{B}^2 \mathbf{C} \mathcal{V}^2} \leq 0,
	 U_{C \mathcal{B}^2 \mathbf{D} \mathcal{V}^2} \ge 0, \\
(c) \ &U_{C \mathcal{B}^2 \mathbf{C} \mathcal{V}^2} \leq 0,
	 U_{D \mathcal{B}^2 \mathbf{D} \mathcal{V}^2} \ge 0,
	 U_{D \mathcal{B}^2 \mathbf{C} \mathcal{V}^2} \ge 0,
	 U_{C \mathcal{B}^2 \mathbf{D} \mathcal{V}^2} \leq 0, \\
(d) \ &U_{C \mathcal{B}^2 \mathbf{C} \mathcal{V}^2} \leq 0,
	 U_{D \mathcal{B}^2 \mathbf{D} \mathcal{V}^2} \ge 0,
	 U_{D \mathcal{B}^2 \mathbf{C} \mathcal{V}^2} \leq 0,
	 U_{C \mathcal{B}^2 \mathbf{D} \mathcal{V}^2} \ge 0, \\
\end{split}
\end{equation}
\end{small}
where $U_{\mathcal{B}^1 \mathcal{B}^2 \mathcal{V}^ 1 \mathcal{V}^2}:=$ $\alpha \bold{S}^1_{\mathcal{B}^1 \mathcal{B}^2 \mathcal{V}^ 1 \mathcal{V}^2} + \beta \bold{S}^2_{\mathcal{B}^1 \mathcal{B}^2 \mathcal{V}^ 1 \mathcal{V}^2} + \gamma \bold{1}$, $\mathcal{B}^1 \mathcal{B}^2 \mathcal{V}^ 1 \mathcal{V}^2 \in \mathcal{S}$,
then a positive or a negative determinant strategy $(\bold{p}^1_{\mathcal{B}}, \bold{p}^1_{\mathcal{V}})$ of player $1$ exists.
\end{theorem}

\noindent \textbf{Proof.}
We consider
\begin{equation}
\bold{U}=\alpha \bold{S}^1 + \beta \bold{S}^2 +\gamma \bold{1}=
	\begin{pmatrix}
	\bold{A_1}  \\
	\bold{A_1} 
	\end{pmatrix}
	 +
	\begin{pmatrix}
	\bold{O}  \\
	\bold{B_1} 
	\end{pmatrix}.
\end{equation}
	We denote the coefficient matrix as 
	\begin{equation*}
	M=\begin{pmatrix}
	\tilde {\bold{p}^1_{\mathcal{B}}}& \tilde {\bold{p}^1_{\mathcal{V}}}
	\end{pmatrix},
	\end{equation*}
	and the augmented matrix as
	\begin{equation*}
	\tilde{M}=\begin{pmatrix}
	\tilde {\bold{p}^1_{\mathcal{B}}} & \tilde {\bold{p}^1_{\mathcal{V}}} & \bold{A_1}
	\end{pmatrix}.
	\end{equation*}
Since all entries of $\bold{p}^1_{\mathcal{B}}$ and $\bold{p}^1_{\mathcal{V}}$ are probabilities, they lie in $[0,1]$. 
Therefore, the entries of $M$ have simple sign properties. 
In particular, the first row entries $h^1_{CC\mathbf{C}}-1$ and $g^1_{CC\mathbf{C}}-1$ have the same sign, while the entries in the third row $h^1_{DC\mathbf{C}}$ and $g^1_{DC\mathbf{C}}-1$ have different signs. 
From this structure, the two columns of $M$ are linearly independent, and we have $r(M)=2$.
Since $\tilde{M}$ is obtained by adding the column $\bold{A}_1$, the rank of $\tilde{M}$ remains 2 if and only if this new column can be expressed in the span of the columns of $M$. 
This is equivalent to Inequality~\ref{inl.1}. 
In other words, if Inequality.\ref{inl.1} is satisfied, then there exist $a, b$, not all zero, such that Eq.\ref{eq.14} or Eq.\ref{eq.15} has non-zero solutions. 
Based on Theorem.\ref{thm1}, a positive or a negative determinant strategy $(\bold{p}^1_{\mathcal{B}}, \bold{p}^1_{\mathcal{V}})$ of player $1$ exists.

\begin{remark}
For any $\mathcal{B}^ i \in \{ C, D\}, \mathcal{V}^ i \in \{ \mathbf{C}, \mathbf{D}\}$, there exist $a$ and $b$, not all zero, such that the Press–Dyson vectors satisfy Eq.\ref{eq.14} or Eq.\ref{eq.15},
if and only if there exist constants $\alpha, \beta, \gamma$, not all zero, which are unilaterally decided by the focal player $1$,
such that one of the Inequality.\ref{inl.1} is satisfied.
Thus, Theorem.\ref{thm2} is an equivalent with Theorem.\ref{thm1}.
\end{remark}

\begin{remark}
If both are focal individuals, then each can have their own $\alpha$, $\beta$, $\gamma$, which do not depend on the other player, i.e., control input.
\end{remark}

\begin{corollary}\label{cor1}
For $\varepsilon_1 = \varepsilon_2 = 0$, neither positive determinant strategies nor negative determinant strategies exist, but zero determinant strategies exist.
\end{corollary}
\noindent \textbf{Proof.}
If $\varepsilon_1 = \varepsilon_2 = 0$, then $\bold{U} = (\bold{A_1}^T, \bold{A_1}^T)^T = (\bold{A_2}^T, \bold{A_2}^T)^T$.
Thus, Eqs.\ref{U1}, \ref{U2} are the same as Eq.\ref{spl13}.
According to Lemma.\ref{lem1}, $\bold{v} \cdot \bold{U} = 0$.
Thus, zero determinant strategies exist, but positive determinant strategies and negative 

\begin{corollary}\label{cor2}
(Extortionate positive determinant strategy) For the two-player two-strategy repeated game given by Eq.\ref{game1} with $T > \varepsilon_2 > R > \varepsilon_1 > 0 > P > S$, if there exists a sufficiently large constant $C > 0$, such that $\varepsilon_2 -\varepsilon_1 \ge C$,
then individual $1$ is able to take a positive determinant strategy $(\bold{p}^1_{\mathcal{B}}, \bold{p}^1_{\mathcal{V}})$ to have more payoff than the opponent, that is, $s^1 > \chi s^2$, in which $\chi \ge \max \{ \frac{R}{R-\varepsilon_1}, \frac{P-\varepsilon_2}{P} \}>1$.
\end{corollary}

\noindent \textbf{Proof.}
Individual $1$ unilaterally controls two players' average payoff to fulfill $s^1 - \chi s^2 > 0$.
The affine relation vector $\bold{U} = \bold{S}^1 - \chi \bold{S}^2$ is 
\begin{equation*}
\begin{split}
		& (R (1 - \chi), 
		S - \chi (T - \varepsilon_2), 
		T - \varepsilon_2 -\chi S, 
		(P - \varepsilon_2)(1 - \chi),\\&
		R (1 - \chi) - \varepsilon_1,  
		S - \varepsilon_1 - \chi (T - \varepsilon_2),
		T - \chi S, 
		P (1 - \chi) +\chi \varepsilon_2,\\&
		R (1 - \chi) +\chi \varepsilon_1,
		S - \chi T,
		T  - \varepsilon_2 - \chi (S - \varepsilon_1),
		P (1 - \chi) - \varepsilon_2,\\&
		(R - \varepsilon_1) (1 - \chi),
		S - \varepsilon_1 - \chi T,
		T - \chi (S - \varepsilon_1),
		P (1 - \chi))^T.
\end{split}
\end{equation*}
For $T > \varepsilon_2 > R > \varepsilon_1 > 0 > P > S$ and $\chi \ge \max \{ \frac{R}{R-\varepsilon_1}, \frac{P-\varepsilon_2}{P} \}>1$, we have 
	$U_{C C \mathbf{C} \mathbf{C}} \leq 0$,
	$U_{C D \mathbf{C} \mathbf{C}} \leq 0$,
	$U_{D C \mathbf{C} \mathbf{C}} \ge 0$,
	$U_{D D \mathbf{C} \mathbf{C}} \ge 0$,
	$U_{C C \mathbf{D} \mathbf{C}} \leq 0$,
	$U_{C D \mathbf{D} \mathbf{C}} \leq 0$,
	$U_{D C \mathbf{D} \mathbf{C}} \ge 0$,
	$U_{D D \mathbf{D} \mathbf{C}} \ge 0$,
	$U_{C C \mathbf{C} \mathbf{D}} \leq 0$,
	$U_{C D \mathbf{C} \mathbf{D}} \leq 0$,
	$U_{D C \mathbf{C} \mathbf{D}} \ge 0$,
	$U_{D D \mathbf{C} \mathbf{D}} \ge 0$,
	$U_{C C \mathbf{D} \mathbf{D}} \leq 0$,
	$U_{C D \mathbf{D} \mathbf{D}} \leq 0$,
	$U_{D C \mathbf{D} \mathbf{D}} \ge 0$,
	$U_{D D \mathbf{D} \mathbf{D}} \ge 0$.
It satisfies the condition ii) in the Theorem.\ref{thm2}.
Thus, there exist a positive determinant strategy $(\bold{p}^1_{\mathcal{B}}, \bold{p}^1_{\mathcal{V}})$, such that the focal individual $1$ gets more payoff than the opponent, that is, $s^1 > \chi s^2$, in which $\chi \ge \max \{ \frac{R}{R-\varepsilon_1}, \frac{P-\varepsilon_2}{P} \}>1$.

\begin{remark}
Our extortionate positive determinant strategy given by Corollary.\ref{cor2} is formally more powerful than the extortionate ZD strategy.
Previous work found an extortionate ZD strategy, which is defined in \cite{Press_2012}:
If an individual adopts an extortionate ZD strategy, her increase of payoff exceeds that of the opponent by a fixed percentage $\chi$.
However, if an individual adopts a extortionate positive determinant strategy, her payoff exceeds that of the opponent by a percentage not less than $\chi > 1$.
\end{remark}

\section{Examples and applications in control systems}
\label{sec6}
We give some examples to show how a single focal individual using positive/negative determinant strategies performs payoff control.

\noindent \textbf{Example 1.}\label{eg1}
\textbf{The focal individual unilaterally controls the opponent’s payoff below a given value.}
Focal individual $1$ wants to unilaterally control the average payoff of individual $2$ below the value $1$, that is, $\alpha = 0, \beta = 1, \gamma = -1$.
We give a negative determinant strategy for such payoff control.

We assume $(R, S, T, P, \varepsilon_1, \varepsilon_2)$ = $(3, 0, 5, 1, 0.01, 2)$. 
If the focal individual $1$ adopts a negative determinant strategy:
\begin{equation}\label{example1-1}
\begin{split}
&\bold{p}^1_{\mathcal{B}}= (0.602, 0.2, 0.202, 0, 0.602, 0.2, 0.202, 0), \\ 
&\bold{p}^1_{\mathcal{V}} =(0.5, 0.5, 0.5, 0.5, 0.5, 0.5, 0.5, 0.5),
\end{split}
\end{equation} 
then the average payoff of individual $2$ is always smaller than the value $1$ (i.e., $s^2 < 1$), regardless of the strategies adopted by individual $2$ (See Fig.~\ref{fig2}(a)).
\begin{figure}
	\begin{minipage}{1\linewidth} 
	\centerline{\includegraphics[width=1\linewidth]{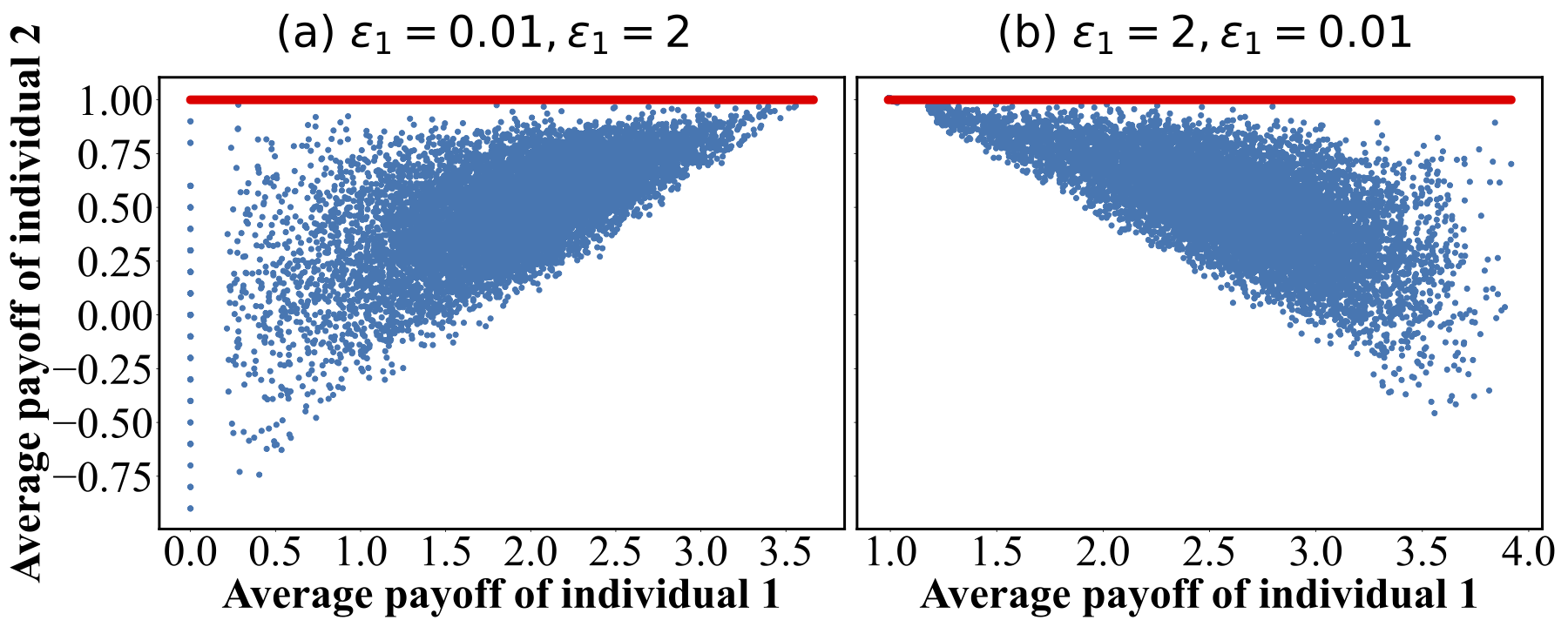}}
	\end{minipage} 
	\caption{The focal individual unilaterally controls the opponent’s payoff below $1$.
	The red lines denote the affine relation $s^2 -1 = 0$, in which $s^2$ is the average payoff of individual $2$. 
	(a) The internal costs are $\varepsilon_1 = 0.01$, $\varepsilon_2 = 2$. 
	Individual $1$ adopts the negative determinant strategy given by Eq.~\ref{example1-1}. 
	(b) The internal costs are $\varepsilon_1 = 2$, $\varepsilon_2 = 0.01$. 
	Individual $1$ adopts the negative determinant strategy given by Eq.~\ref{example1-2}. 
	The red lines denote the affine relation $s^2 -1 = 0$. 
	The blue dots indicate the average payoffs for both individuals
	, where strategies ($\bold{p}^2_{\mathcal{B}}, \bold{p}^2_{\mathcal{V}}$) of individual $2$ are sampled $10^4$ times. 
	For each strategy pair ($\bold{p}^1_{\mathcal{B}}, \bold{p}^1_{\mathcal{V}}$, $\bold{p}^2_{\mathcal{B}}, \bold{p}^2_{\mathcal{V}}$), we compute the average payoffs of both individuals using the determinant formula provided in Eq.\ref{payoff1}.
	Parameters: $(R, S, T, P) = (3, 0, 5, 1)$.}\label{fig2}
\end{figure}

If $\varepsilon_1 = 2, \varepsilon_2= 0.01$. 
We find a negative determinant strategy of focal individual $1$ that controls the average payoff of individual $2$ is always smaller than the value $1$, regardless of the strategies adopted by individual $2$ (See Fig.~\ref{fig2}(b)). The strategy is given by: 
\begin{equation}\label{example1-2}
\begin{split}
&\bold{p}^1_{\mathcal{B}}= (0.6, 0.202, 0.2, 0.002, 0.6, 0.202, 0.2, 0.002), \\ 
&\bold{p}^1_{\mathcal{V}} =(0.5, 0.5, 0.5, 0.5, 0.5, 0.5, 0.5, 0.5).
\end{split}
\end{equation} 

\noindent \textbf{Example 2.}\label{eg2}
\textbf{The focal individual unilaterally controls her own increasement of payoff more than that of the opponent.}
	The focal individual $1$ unilaterally controls two players' average payoff to fulfill $s^1 - k > \chi (s^2 - k)$, in which $\chi \ge 1$ denotes the extortion factor and $k$ denotes the baseline of extortion. 
	The extortion factor $\chi$ measures the fixed percentage that individual $1$'s increase of payoff exceeds that of $2$, and the baseline $k$ affects the uncertainty of extortion \cite{Hao_2015}.
Focal individual $1$ wants to unilaterally control two players' average payoff to fulfill $s^1 - 1.5 > 3 (s^2 - 1.5)$. (In this case, we set $\chi = 3, k = 1.5$, which is equivalent to $\alpha = 1, \beta = -3, \gamma = 3$).
We give a positive determinant strategy for such payoff control.

We assume $(R, S, T, P, \varepsilon_1, \varepsilon_2)$ = $(3, 0, 5, 1, 0.01, 0.9)$. 
If individual $1$ uses a positive determinant strategy 
\begin{equation}\label{example2-1}
\begin{split}
\bold{p}^1_{\mathcal{B}} &= (0.8515, 0.4, 0.3565, 0.005, 0.851, 0.3995, 0.4015, 0.05), \\ 
\bold{p}^1_{\mathcal{V}} &= (0.5, 0.5, 0.5, 0.5, 0.5, 0.5, 0.5, 0.5),\notag
\end{split}
\end{equation} 
then average payoffs of two players fulfill $s^1 - 1.5 > 3 (s^2 - 1.5)$, regardless of the strategies adopted by individual $2$ (See Fig.~\ref{fig3}(a)).
\begin{figure}
	\begin{minipage}{1\linewidth} 
	\includegraphics[width=1\linewidth]{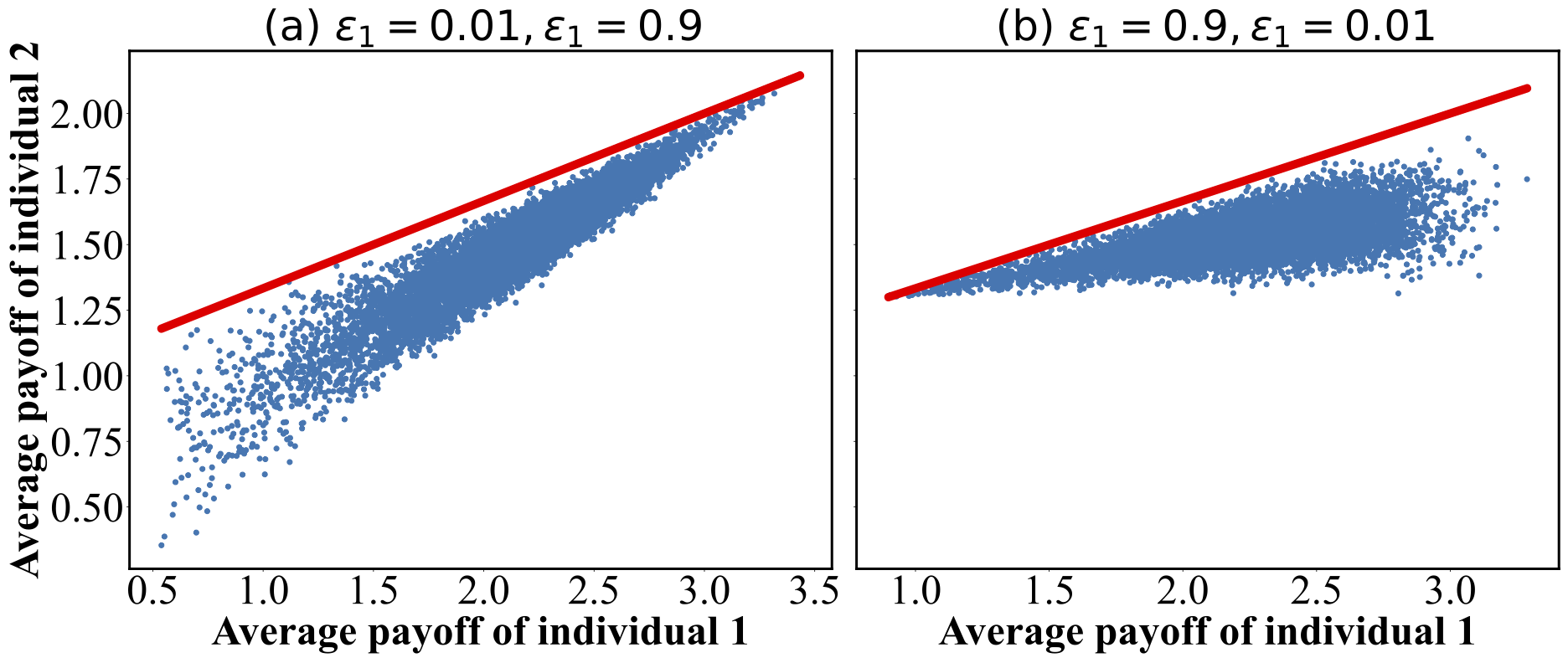}
	\end{minipage}
	\caption{\label{fig3} 
	Individual $1$ unilaterally controls two players' average payoff fulfilling $s^1 - 1.5 > 3 (s^2 - 1.5)$, where $s^1, s^2$ are the average payoffs of individuals $1$ and $2$, respectively. 
	(a) The internal costs are $\varepsilon_1 = 0.01$, $\varepsilon_2 = 0.09$. 
	(b) The internal costs are $\varepsilon_1 = 0.09$, $\varepsilon_2 = 0.01$. 
	The red lines denote the affine relation $s^1 -1.5 = 3 (s^2 - 1.5)$. 
	The blue dots indicate the average payoffs for both individuals.
	Parameters: $(R, S, T, P) = (3, 0, 5, 1)$.}
\end{figure}

We assume $\varepsilon_1 = 0.9$ and $\varepsilon_2 = 0.01$.
If individual $1$ uses a positive determinant strategy
\begin{equation}\label{example2-2}
\begin{split}
\bold{p}^1_{\mathcal{B}} &= (0.85, 0.4015, 0.3995, 0.051, 0.805, 0.3565, 0.4, 0.0515), \\ 
\bold{p}^1_{\mathcal{V}}  &=(0.5, 0.5, 0.5, 0.5, 0.5, 0.5, 0.5, 0.5),\notag
\end{split}
\end{equation} 
then two players' average payoff fulfill $s^1 - 1.5 > 3 (s^2 - 1.5)$, regardless of the strategies adopted by individual $2$ (See Fig.~\ref{fig3}(b)).

\noindent \textbf{Example 3.}\label{eg3}
\textbf{The focal individual unilaterally controls her payoff more than that of the opponent.}
	The focal individual $1$ unilaterally controls two players' average payoff fulfilling $s^1 > \chi s^2$, in which $\chi \ge \max \{ \frac{R}{R-\varepsilon_1}, \frac{P-\varepsilon_2}{P} \} > 1$ denotes the extortion factor. 
Focal individual $1$ wants to unilaterally control two players' average payoff fulfilling $s^1 > 5 s^2$. (In this case, it is equivalent to $\alpha = 1, \beta = -5, \gamma = 0$).

We consider a case with parameters $(R, S, T, P, \varepsilon_1, \varepsilon_2) = (2, -3, 6, -1, 0.01, 4)$ according to Corollary.\ref{cor2}.
We find a positive determinant strategy of focal individual $1$ 
\begin{small}
\begin{equation}\label{example3}
\begin{split}
\bold{p}^1_{\mathcal{B}}&= (0.841, 0.34, 0.341, 0, 0.8408, 0.3398, 0.421, 0.08), \\ 
\bold{p}^1_{\mathcal{V}}&= (0.5, 0.5, 0.5, 0.5, 0.5, 0.5, 0.5, 0.5).
\end{split}
\end{equation} 
\end{small}
Then two players' average payoff fulfill $s^1 > 5 s^2$, regardless of the strategies adopted by individual $2$ (See Fig.~\ref{fig4}(a)).
Furthermore, for individual $1$ to adopt strategy given by Eq.~\ref{example3}, her payoff must be five times more than individual $2$'s payoff, and there is even a 0.99 probability that it is $100$ times more than individual $2$'s payoff (See Fig.~\ref{fig4}(b)).
Thus, the extortionate positive determinant strategy is formally more powerful.
\begin{figure}
	\centering
	\begin{minipage}{1\linewidth} 
	\includegraphics[width=1\linewidth]{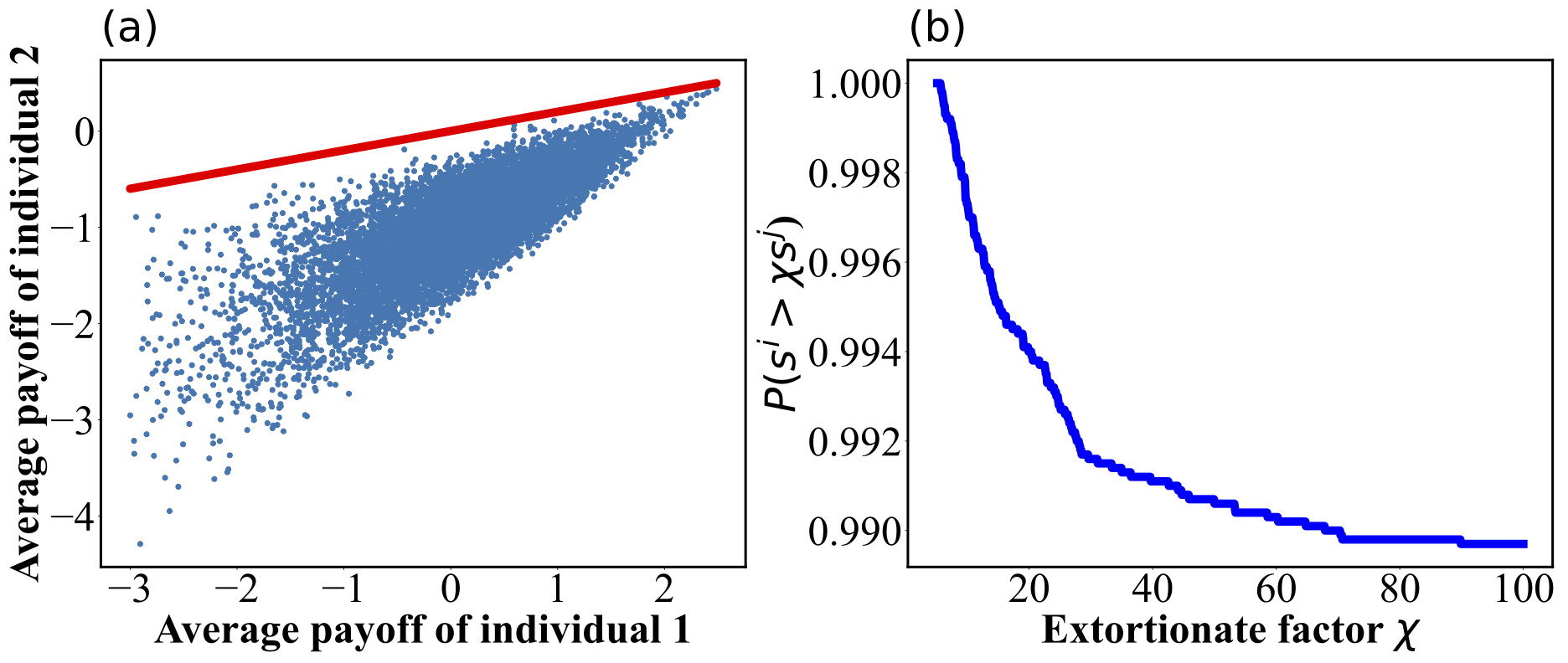}
	\end{minipage}
	\caption{\label{fig4} 
	(a) Individual $1$ unilaterally controls her payoff more than five times of that of the opponent, i.e., $s^1 > 5 s^2$. 
	The red line denotes the affine relation $s^1=  5 s^2$. 
	(b) The probability of $s^1 > \chi s^2$ ($5 < \chi < 100$) is more than $0.99$.
	It implies that for individual $1$ with strategy given by Eq.~\ref{example3}, her payoff must be five times more than individual $2$'s payoff, and there is even a probability of more than $0.99$ that it would be $100$ times the payoff of an individual $2$.
	In two panels, focal individual $1$ adopts the positive determinant strategy given by Eq.~\ref{example3}. 
	The strategies ($\bold{p}^2_{\mathcal{B}}, \bold{p}^2_{\mathcal{V}}$) of individual $2$ are sampled $10^4$ times. 
	For each strategy pair ($\bold{p}^1_{\mathcal{B}}, \bold{p}^1_{\mathcal{V}}$, $\bold{p}^2_{\mathcal{B}}, \bold{p}^2_{\mathcal{V}}$), we compute the average payoffs of both individuals using the determinant formula provided in Eq.\ref{payoff1}.
Parameters: $(R, S, T, P, \varepsilon_1, \varepsilon_2) = (2, -3, 6, -1, 0.01, 4)$.
}
	\vspace*{-5pt}
\end{figure}

Here, both individuals act simultaneously using stochastic memory-one strategies (all entries are between $0$ and $1$), allowing each to be the focal individual in our framework. 
Here, we show how payoff control is achieved by both individuals using positive/negative determinant strategies simultaneously.

\noindent \textbf{Example 4.}
\textbf{Both players have the positive determinant strategies.}
We consider a game with $(R, S, T, P, \varepsilon_1, \varepsilon_2)$ = $(3, 0, 5, 1, 0.01, 0.9)$. Individual $1$ wants to unilaterally control the average payoffs to satisfy $s^1 – 3 s^2 + 3 > 0,$ and individual $2$ wants to unilaterally control the average payoffs to satisfy $-3  s^1 + s^2+ 3 > 0$.
In this case, the individual $1$ adopts the strategy:
\begin{equation*}
\begin{split}
\bold{p}^1_{\mathcal{B}} = (&0.8020, 0.2000, 0.4753, 0.0067, 0.8013, 0.1993, 0.5353, 0.0667),\\
\bold{p}^1_{\mathcal{V}}= (&0.5, 0.5, 0.5, 0.5, 0.5, 0.5, 0.5, 0.5).\\
\end{split}
\end{equation*}
And the individual $2$ adopts the strategy:
\begin{equation*}
\begin{split}
\bold{p}^2_{\mathcal{B}} = (&0.8020, 0.2000, 0.4753, 0.0067, 0.8013, 0.1993, 0.5353, 0.0667),\\
\bold{p}^2_{\mathcal{V}}= (&0.5, 0.5, 0.5, 0.5, 0.5, 0.5, 0.5, 0.5).\\
\end{split}
\end{equation*}
Then the average payoff of first individual $s^1 = 0.8090$, and that of the second individual $s^2 = 0.8594$. It satisfies that $s^1 – 3 s^2 + 3 = 1.2308 > 0$ and $-3  s^1 + s^2+ 3 = 1.4323 > 0$. 
Thus, the strategies are positive determinant strategies for both individuals.

\noindent \textbf{Example 5.}\label{eg5}
\textbf{One of the players has the positive determinant strategy, while the opponent has the negative determinant strategy.}
We still consider parameters $(R, S, T, P, \varepsilon_1, \varepsilon_2)$ = $(3, 0, 5, 1, 0.01, 0.9)$. Individual $1$ wants to unilaterally control the average payoffs to satisfy $s^1 – 3 s^2 + 3 < 0,$ and individual $2$ wants to unilaterally control the average payoffs to satisfy $-3  s^1 + s^2+ 3 > 0$.
In this case, the individual $1$ adopts the strategy:
\begin{equation*}
\begin{split}
\bold{p}^1_{\mathcal{B}} = (&0.8, 0.3800, 0.4733, 0.1867, 0.7993, 0.3793, 0.5333, 0.2467);\\
\bold{p}^1_{\mathcal{V}}= (&0.5, 0.5, 0.5, 0.5, 0.5, 0.5, 0.5, 0.5).\\
\end{split}
\end{equation*}
And the individual $2$ adopts the strategy:
\begin{equation*}
\begin{split}
\bold{p}^2_{\mathcal{B}} = (&0.8020, 0.2000, 0.4753, 0.0067, 0.8013, 0.1993, 0.5353, 0.0667);\\
\bold{p}^2_{\mathcal{V}}= (&0.5, 0.5, 0.5, 0.5, 0.5, 0.5, 0.5, 0.5).\\
\end{split}
\end{equation*}
Then the average payoff of first individual $s^1 = 1.1993$, and that of the second individual $s^2 = 1.7362$. It satisfies that $s^1 – 3 s^2 + 3 = -1.0092 < 0$ and $-3  s^1 + s^2+ 3 = 1.1383 > 0$. 
Thus, the strategies of individuals $1$ and $2$ are negative determinant strategy and positive determinant strategy, respectively.

\section{Conclusion}
\label{sec5}
We have established a two-player repeated game by assuming that an individual pays the internal cost if her behavior is not consistent with her internal thought. 
We have shown that positive and negative determinant strategies are able to unilaterally control the payoffs.
It provides an efficient way to payoff control although the classical ZD strategies is absent, which opens an avenue for payoff control in a general setting.
Payoff control is more often refereed as Zero-determinant strategy (ZD strategy) \cite{hao2018payoff, shi2025payoff, Govaert2021}. 
This is because it is obtained via making the determinant of a matrix zero \cite{Press_2012}. 
The reason why we do not use this terminology (zero-determinant strategy) is that the determinant in our work is not zero.
Instead, we use payoff control.
We also have shown that the newly found extortionate positive determinant strategies are more powerful for payoff control than extortionate ZD strategies.

Our proposed game is a generalization of previous works.
In classical cooperative game setting, the action gives rise to the strategy set. This is also true in existing literatures on repeated games. 
In our model, a strategy refers to a pair $(X, Y)$, where $X$ is the behavior and $Y$ is the internal thought. 
This leads to a larger strategy set. 
Consequently, individuals can have different payoffs even though they choose the same action.
This difference can arise from
 spatial structure \cite{nowak_1992, ohtsuki_2006},
 strategy diversity \cite{archetti2012game, szolnoki2016competition} and vary environments \cite{weitz_2016, su2019evolutionary}.
Here, the difference arises from inconsistency between internal thoughts and actions \cite{Traulsen_2023, Liu_2022}.
In fact, under social pressure or other psychological factors, individuals can do what they do not want to, or they do not do what they want to \cite{Asch_1956, Schwarz_2000, 10.2307/24936719}.
Thus, the inconsistency is widespread in systems with highly intelligent agents.
On the other hand, our model also explicitly takes noise into account for decision making.
Noteworthily, the noise is not in the payoff matrix, but in the strategy updating. 
For example, the tit-for-tat strategy involves cooperating in the first round and then exactly repeating the opponent's previous behavior, with reactive probabilities being $0$ or $1$. The noisy tit-for-tat strategy introduces probabilities (neither $0$ nor $1$), accounting for potential errors. 
However, the noise-induced strategy profile is much more complicated for our proposed game, since the strategy set of our proposed game is larger than that of the classical games.

Technically, our work offers new insights to payoff control beyond classical zero determinant strategies.
After introducing non-zero internal costs, we split the affine combination of payoffs into two parts, one of which refers to Press-Dyson vectors which is closely related to the ZD strategy, whereas the other of which vanishes in ZD strategies. 
The name positive determinant/negative determinant strategy is the summary to show how technically different our strategy differs from zero-determinant strategy. 
This technical improvement paves the way for payoff control, even when ZD strategy is absent. 
Therefore, our analysis helps design payoff control strategies in systems where agents have values besides actions (including humans and AI). 

The novelty of our work has two sides compared with the literatures in the past decades:
One is that we introduce value into our model, creating a more complex game, in which individuals have both values and actions together with the inconsistency cost in the utility function.
Noteworthily, this framework differs from uncertain payoffs in Bayesian game.
In our setting, individuals are sure about their own value (cooperation or not), they know nothing about other’s value and they do not even try to know. 
The other novelty lies in identifying a time-invariant strategy for unilateral payoff control in such complex behavior-value games.
And we also prove that the classical payoff control strategies do not exist in such a complex game scenario.
The existence of such strategy is not trivial, but what is even more non-trivial is the time-invariance of the strategy. 
There is no need to adjust the individual's payoff to motivate her to achieve the goal (e.g., utility design), nor is there a need to optimize the path to the goal as the number of game rounds increases (e.g., reward shaping).

Our theoretical framework holds promise for future inventions, including but not limited to the following research directions: 
i) extending two-player two-strategy repeated game to multiplayer multi strategy repeated games;
ii) considering more than three internal thoughts and assessing their resilience to uncertainty and noise in payoff environments;
iii) analyzing the evolutionary stability of the positive/negative strategies.
iv) designing reward and punishment mechanisms that adjust the individual’s payoff in each round based on their behavior and values, in order to explore how time-invariant strategies influence long-term outcomes;
v) examining settings in which the focal individual adopts the proposed time-invariant strategy, while the opponent follows a specific learning strategy (e.g., reinforcement learning) in each round;
vi) integrating Bayesian game theory could offer further insights, especially in scenarios where agents try to estimate each other's internal values and update these beliefs over time.

To sum up, we have shown a class of positive (negative) determinant strategies for games with behavior-value inconsistency, which is more powerful than extortionate ZD strategy.
This opens up an avenue for the coevolution of behavior and inner state.

\nocite{*}

\bibliographystyle{plain}
\bibliography{main}

\end{document}